\documentclass{revtex4-2}
\usepackage{amsmath,amssymb,amsthm,graphicx,bbm,braket}
\usepackage[
	bookmarks = false,
	colorlinks = true,
	linkcolor = blue,
	urlcolor= black,
	citecolor = blue]{hyperref}
\newtheorem{theorem}{Theorem}

\newtheorem{corollary}{Corollary}
\newtheorem{definition}{Definition}

\DeclareMathOperator{\Tr}{Tr}

\newcommand{\eref}[1]{\eqref{#1}}
\newcommand{\Eref}[1]{Equation~\eqref{#1}}
\newcommand{\Sref}[1]{Section~\ref{#1}}
\newcommand{\Fref}[1]{Figure~\ref{#1}}
\begin{document}
%%%%%%%%%%%%%%%%%%%%%%
%       Title        %
%%%%%%%%%%%%%%%%%%%%%%
\title{Entanglement witness and multipartite quantum state discrimination}
\author{Donghoon Ha}
\affiliation{Department of Applied Mathematics and Institute of Natural Sciences, Kyung Hee University, Yongin 17104, Republic of Korea}
\author{Jeong San Kim}
\email{freddie1@khu.ac.kr}
\affiliation{Department of Applied Mathematics and Institute of Natural Sciences, Kyung Hee University, Yongin 17104, Republic of Korea}
%%%%%%%%%%%%%%%%%%%%%%
%      Abstract      %
%%%%%%%%%%%%%%%%%%%%%%
\begin{abstract}
We consider multipartite quantum state discrimination and show that the minimum-error discrimination by separable measurements is closely related to the concept of entanglement witness. Based on the properties of entanglement witness, we establish some necessary and/or sufficient conditions on minimum-error discrimination by separable measurements. We also provide some conditions on the upper bound of the maximum success probability over all possible separable measurements. Our results are illustrated by examples of multidimensional multipartite quantum states. Finally, we provide a systematic way in terms of the entanglement witness to construct multipartite quantum state ensembles showing nonlocality in state discrimination.
\end{abstract}
\maketitle
%%%%%%%%%%%%%%%%%%%%%%
%    Introduction    %
%%%%%%%%%%%%%%%%%%%%%%
\section{Introduction}
%%%%%%%%%%%%%%%%%%%%%%
Quantum state discrimination is one of the fundamental concepts used in various quantum information and computation theory\cite{chef2000,barn20091,berg2010,bae2015}. 
In general, we can always perfectly discriminate orthogonal quantum states using appropriate measurement. However, nonorthogonal quantum states cannot be perfectly discriminated by means of any measurement. For this reason, various state discrimination strategies have been studied for optimal discrimination of nonorthogonal quantum states, such as minimum-error discrimination, unambiguous discrimination and maximum-confidence discrimination\cite{hels1969,ivan1987,diek1988,pere1988,crok2006}.
%%%%%%%%%%%%%%%%%%%%%%

%%%%%%%%%%%%%%%%%%%%%%
\emph{Entanglement witness}(EW) is an important tool to
detect the existence of entanglement inherent in a multipartite quantum state\cite{horo1996,terh2000,lewe2000,chru2014}. Mathematically, EW is a Hermitian operator having non-negative mean value for every separable state, but negative for some entangled states.
As EW provides an useful methodology to detect entanglement that is an important quantum nonlocality, 
it is natural to ask whether EW can also be used to characterize 
other nonlocal phenomenon of multipartite quantum states.
%%%%%%%%%%%%%%%%%%%%%%

%%%%%%%%%%%%%%%%%%%%%%
Quantum nonlocal phenomenon also arises in discriminating multipartite quantum states; quantum nonlocality occurs when optimal state discrimination cannot be realized only by \emph{local operations and classical communication}(LOCC)\cite{pere1991,benn19991,ghos2001,chit2013,jian2020,zuo2021,wang2023}. 
However, characterizing local discrimination of quantum states is a hard task and very little is known due to the lack of good mathematical structure for LOCC.
%%%%%%%%%%%%%%%%%%%%%%

%%%%%%%%%%%%%%%%%%%%%%
Here, we establish a specific relation between the properties of EW and separable measurements, a mathematically well-structured set of measurements having LOCC measurements as a special case.
We show that the minimum-error discrimination of multipartite quantum states using separable measurements strongly depends on the existence of EW.
More precisely, we establish conditions on minimum-error discrimination by separable measurements in terms of EW. We also provide conditions on the upper bound of the maximum success probability over all possible separable measurements. We illustrate our results using examples of multidimensional multipartite quantum states. Finally, we provide a systematic way in terms of EW to construct multipartite quantum state ensembles showing nonlocality in state discrimination.
%%%%%%%%%%%%%%%%%%%%%%

%%%%%%%%%%%%%%%%%%%%%%
This paper is organized as follows.
In \Sref{sec:pre},
we first recall the definitions and some properties about separable measurements and EW. We also recall the definition of minimum-error discrimination as well as some useful properties of the optimal measurements. 
In \Sref{sec:ppq}, we provide conditions on the upper bound of the maximum success probability over all possible separable measurements (Theorems~\ref{thm:pptq} and \ref{thm:mnsc}).
In \Sref{sec:mep},
we provide conditions on minimum-error discrimination by separable measurements in terms of EW (Theorems~\ref{thm:qmsc} and \ref{thm:qupb}). Our results are illustrated by examples of multidimensional multipartite quantum states.
In \Sref{sec:ces}, we provide a systematic way in terms of EW to construct multipartite quantum state ensembles showing nonlocality in state discrimination. We conclude our results in \Sref{sec:dis}.
%%%%%%%%%%%%%%%%%%%%%%

%%%%%%%%%%%%%%%%%%%%%%
%      Section       %
%%%%%%%%%%%%%%%%%%%%%%
\section{Preliminaries}\label{sec:pre}
%%%%%%%%%%%%%%%%%%%%%%
For a multipartite Hilbert space $\mathcal{H}=\bigotimes_{k=1}^{m}\mathbb{C}^{d_{k}}$ with $m\geqslant2$ and positive integers $d_{k}$ for $k=1,\ldots,m$, let $\mathbb{H}$ be the set of all Hermitian operators acting on $\mathcal{H}$.
We use $\mathbb{H}_{+}$ to denote the set of all positive-semidefinite operators in $\mathbb{H}$, that is,
\begin{equation}\label{eq:H+} 
\mathbb{H}_{+}=\{E\in\mathbb{H}\,|\,\bra{v}E\ket{v}\geqslant0~~\forall \ket{v}\in\mathcal{H}\}.
\end{equation}
A multipartite quantum state is expressed by a density operator $\rho$, that is, $\rho\in\mathbb{H}_{+}$ with $\Tr\rho=1$.
A measurement is represented by a positive operator-valued measure $\{M_{i}\}_{i}$, that is, $\{M_{i}\}_{i}\subseteq\mathbb{H}_{+}$ satisfying $\sum_{i}M_{i}=\mathbbm{1}$, where $\mathbbm{1}$ is the identity operator in $\mathbb{H}$. 
For the quantum state $\rho$, the probability of obtaining the measurement outcome corresponding to $M_{j}$ is $\Tr(\rho M_{j})$.
%%%%%%%%%%%%%%%%%%%%%%

%%%%%%%%%%%%%%%%%%%%%%%
\begin{figure}[!tt]
\centerline{\includegraphics*[bb=0 0 530 525,scale=0.40]{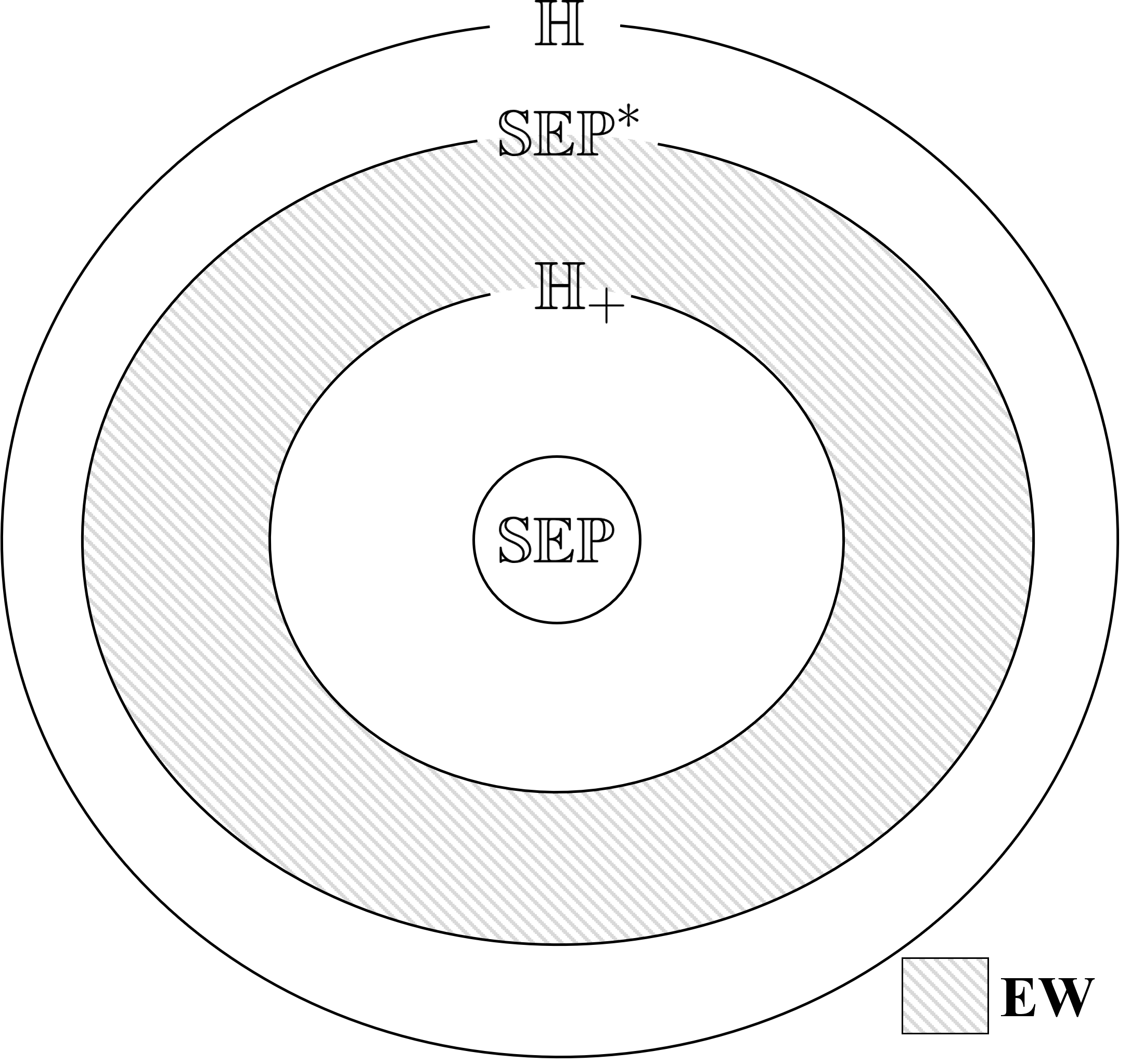}}
\caption{The relationship of the subsets of $\mathbb{H}$. 
The shaded area $\mathbb{SEP}^{*}\setminus\mathbb{H}_{+}$ is the set of all EWs.
}\label{fig:inc}
\end{figure}
%%%%%%%%%%%%%%%%%%%%%%%

%%%%%%%%%%%%%%%%%%%%%%
%     Subsection     %
%%%%%%%%%%%%%%%%%%%%%%
\subsection{Entanglement witness}
%%%%%%%%%%%%%%%%%%%%%%
\begin{definition}\label{def:sep}
$E\in\mathbb{H}_{+}$ is called \emph{separable} if it can represented as a sum of positive-semidefinite product operators, that is,
\begin{equation}\label{eq:sepod}
E=\sum_{l}\bigotimes_{k=1}^{m}E_{l,k},
\end{equation}
where $\{E_{l,k}\}_{l}$ is a set of positive-semidefinite operators acting on $\mathbb{C}^{d_{k}}$ of $\mathcal{H}$ for each $k=1,\ldots,m$.
Similarly, we say that the set $\{E_{i}\}_{i}\subseteq\mathbb{H}_{+}$ is \emph{separable} if $E_{i}$ is separable for all $i$.
\end{definition}
%%%%%%%%%%%%%%%%%%%%%%
\noindent We denote the set of all \emph{separable} operators in $\mathbb{H}_{+}$ as
\begin{equation}\label{eq:sepdef}
\mathbb{SEP}=\{E\in\mathbb{H}_{+}\,|\, E:\mbox{separable}\},
\end{equation}
and its dual set as $\mathbb{SEP}^{*}$, that is,
\begin{equation}\label{eq:sepsd}
\mathbb{SEP}^{*}=\{E\in\mathbb{H}\,|\,\Tr(EF)\geqslant0~\forall F\in\mathbb{SEP}\}.
\end{equation}
A Hermitian operator in $\mathbb{SEP}^{*}$ is also called \emph{block positive}.
Because $\mathbbm{1}\in\mathbb{SEP}$, we have
\begin{equation}\label{eq:trez}
\Tr E\geqslant0
\end{equation}
for all $E\in\mathbb{SEP}^{*}$,
where the equality holds if and only if $E$ is the zero operator in $\mathbb{H}$ \cite{ha2022}.
%%%%%%%%%%%%%%%%%%%%%%

%%%%%%%%%%%%%%%%%%%%%%
\begin{definition}\label{def:ew}
$W\in\mathbb{H}$ is called an \emph{EW}
if it is block positive but not positive semidefinite, that is,
\begin{equation}\label{eq:wsshp}
W\in\mathbb{SEP}^{*}\setminus\mathbb{H}_{+}.
\end{equation}
\end{definition}
%%%%%%%%%%%%%%%%%%%%%%
\noindent For each EW $W$, it is known that there exists an entangled operator $E\in\mathbb{H}_{+}\setminus\mathbb{SEP}$ satisfying
\begin{equation}\label{eq:ewep}
\Tr(WE)<0,
\end{equation} 
thus the term \emph{entanglement witness}.
\Fref{fig:inc} illustrates the relationship of the subsets of $\mathbb{H}$.
%%%%%%%%%%%%%%%%%%%%%%

%%%%%%%%%%%%%%%%%%%%%%
%     Subsection     %
%%%%%%%%%%%%%%%%%%%%%%
\subsection{Minimum-error discrimination of multipartite quantum states}
%%%%%%%%%%%%%%%%%%%%%%
For a \emph{multipartite} quantum state ensemble,
\begin{equation}\label{eq:ens}
\mathcal{E}=\{\eta_{i},\rho_{i}\}_{i=1}^{n},
\end{equation}
let us consider the situation that
the state $\rho_{i}$ is prepared with the probability $\eta_{i}$.
To guess the prepared state,
we use the decision rule in terms of
a measurement $\{M_{i}\}_{i=1}^{n}$ such that the detection of $M_{i}$ means that 
the prepared state is guessed to be $\rho_{i}$.
The average probability of correctly guessing the prepared state from $\mathcal{E}$ in Eq.~\eref{eq:ens} with respect to a measurement $\{M_{i}\}_{i=1}^{n}$ is
\begin{equation}\label{eq:apcg}
\sum_{i=1}^{n}\eta_{i}\Tr(\rho_{i}M_{i}).
\end{equation}
%%%%%%%%%%%%%%%%%%%%%%

%%%%%%%%%%%%%%%%%%%%%%
The \emph{minimum-error discrimination} of $\mathcal{E}$ is to achieve the optimal success probability,
\begin{equation}\label{eq:pgdef}
p_{\rm G}(\mathcal{E})=\max_{\rm Measurement}\sum_{i=1}^{n}\eta_{i}\Tr(\rho_{i}M_{i}),
\end{equation}
where the maximum is taken over all possible measurements \cite{hels1969}.
We note that a measurement $\{M_{i}\}_{i=1}^{n}$ provides the optimal success probability $p_{\rm G}(\mathcal{E})$ if and only if it satisfies the following condition \cite{hole1974,yuen1975,barn20092,bae2013},
\begin{equation}\label{eq:nscfme}
\sum_{j=1}^{n}\eta_{j}\rho_{j}M_{j}-\eta_{i}\rho_{i}
\in\mathbb{H}_{+}~\forall i=1,\ldots,n.
\end{equation}
%%%%%%%%%%%%%%%%%%%%%%

%%%%%%%%%%%%%%%%%%%%%%
A measurement $\{M_{i}\}_{i}$ is called a \emph{separable measurement} if $\{M_{i}\}_{i}\subseteq\mathbb{SEP}$, and a measurement is called a \emph{LOCC measurement} if it can be realized by LOCC. It is known that every LOCC measurement is a separable measurement \cite{chit20142}.
When the available measurements are limited to separable measurements, 
we denote the maximum success probability by
\begin{equation}\label{eq:pptdef}
p_{\rm SEP}(\mathcal{E})=\max_{\substack{\rm Separable\\ \rm measurement}}\sum_{i=1}^{n}\eta_{i}\Tr(\rho_{i}M_{i}),
\end{equation}
where the maximum is taken over all possible separable measurements.
Similarly, we denote 
\begin{equation}\label{eq:pldef}
p_{\rm L}(\mathcal{E})=\max_{\substack{\rm LOCC\\ \rm measurement}}\sum_{i=1}^{n}\eta_{i}\Tr(\rho_{i}M_{i}),
\end{equation}
where the maximum is taken over all possible LOCC measurements.
From the definitions, 
we trivially have
\begin{equation}\label{eq:plptpg}
p_{\rm L}(\mathcal{E})\leqslant p_{\rm SEP}(\mathcal{E})\leqslant p_{\rm G}(\mathcal{E}).
\end{equation}
%%%%%%%%%%%%%%%%%%%%%%

%%%%%%%%%%%%%%%%%%%%%%
%      Section       %
%%%%%%%%%%%%%%%%%%%%%%
\section{Separable measurements and quantum state discrimination}\label{sec:ppq}
%%%%%%%%%%%%%%%%%%%%%%
For a multipartite quantum state ensemble $\mathcal{E}$ in Eq.~\eref{eq:ens}, we define $\mathbb{H}_{\rm SEP}(\mathcal{E})$ as
\begin{equation}\label{eq:ses}
\mathbb{H}_{\rm SEP}(\mathcal{E})=\{H\in\mathbb{H}\,|\, 
H-\eta_{i}\rho_{i}\in\mathbb{SEP}^{*}~\forall i=1,\ldots,n\},
\end{equation}
where $\mathbb{SEP}^{*}$ is defined in Eq.~\eref{eq:sepsd}.
In other words, $\mathbb{H}_{\rm SEP}(\mathcal{E})$ is the set of all $H$ such that $H-\eta_{i}\rho_{i}$ is block positive in $\mathbb{H}$, for all $i=1,\ldots,n$. 
We further define
\begin{equation}\label{eq:hppte}
\mathbb{H}_{\rm EW}(\mathcal{E})=
\{H\in\mathbb{H}_{\rm SEP}(\mathcal{E})\,|\,H-\eta_{j}\rho_{j}\notin\mathbb{H}_{+}~\mbox{for some}~j\in\{1,\ldots,n\}\},
\end{equation}
that is, $\mathbb{H}_{\rm EW}(\mathcal{E})$ is the set of all $H\in\mathbb{H}_{\rm SEP}(\mathcal{E})$ such that $H-\eta_{j}\rho_{j}$ is an EW for some $j\in\{1,\ldots,n\}$.
From Definition~\ref{def:ew}, we can see that 
\begin{equation}\label{eq:psmd}
H\in\mathbb{H}_{\rm SEP}(\mathcal{E})\setminus\mathbb{H}_{\rm EW}(\mathcal{E})
\end{equation}
if and only if
\begin{equation}\label{eq:eqhpt}
H-\eta_{i}\rho_{i}\in\mathbb{H}_{+}~\forall i=1,\ldots,n.
\end{equation}
%%%%%%%%%%%%%%%%%%%%%%

%%%%%%%%%%%%%%%%%%%%%%
Now, let us consider the minimum quantity
\begin{equation}\label{eq:qptdef}
q_{\rm SEP}(\mathcal{E})=\min_{H\in\mathbb{H}_{\rm SEP}(\mathcal{E})}\Tr H,
\end{equation}
which is an upper bound of $p_{\rm SEP}(\mathcal{E})$\cite{band2015}, that is,
\begin{equation}\label{eq:pptq}
p_{\rm SEP}(\mathcal{E})\leqslant q_{\rm SEP}(\mathcal{E}).
\end{equation}
The following theorem shows that 
$p_{\rm SEP}(\mathcal{E})$ in Eq.~\eref{eq:pptdef}
is equal to $q_{\rm SEP}(\mathcal{E})$ in Eq.~\eref{eq:qptdef}. 
%%%%%%%%%%%%%%%%%%%%%%

%%%%%%%%%%%%%%%%%%%%%%
%      Theorem       %
%%%%%%%%%%%%%%%%%%%%%%
\begin{theorem}\label{thm:pptq}
For a multipartite quantum state ensemble $\mathcal{E}=\{\eta_{i},\rho_{i}\}_{i=1}^{n}$, 
\begin{equation}\label{eq:ubppt}
p_{\rm SEP}(\mathcal{E})= q_{\rm SEP}(\mathcal{E}).
\end{equation}
\end{theorem}
%%%%%%%%%%%%%%%%%%%%%%

%%%%%%%%%%%%%%%%%%%%%%
\begin{proof}
As we already have Inequality~\eref{eq:pptq}, it is enough to show that 
\begin{equation}\label{eq:pqgt}
p_{\rm SEP}(\mathcal{E})\geqslant q_{\rm SEP}(\mathcal{E}).
\end{equation}
Let us consider the set,
\begin{equation}\label{eq:ase}
\mathcal{S}(\mathcal{E})=
\Big\{\big(\sum_{i=1}^{n}\eta_{i}\Tr(\rho_{i}M_{i})-p,\mathbbm{1}-\sum_{i=1}^{n}M_{i}\big)\in\mathbb{R}\times\mathbb{H}\,\Big|\,p>p_{\rm SEP}(\mathcal{E}),~
M_{i}\in\mathbb{SEP}~\forall i=1,\ldots,n\Big\},
\end{equation}
where $\mathbb{R}$ is the set of all real numbers.
We note that $\mathcal{S}(\mathcal{E})$ is 
a convex set due to the convexity of $\mathbb{SEP}$ in Eq.~\eref{eq:sepdef}.
Moreover, $\mathcal{S}(\mathcal{E})$ does not have the origin $(0,0_{\mathbb{H}})$ of $\mathbb{R}\times\mathbb{H}$ 
otherwise there is a separable measurement $\{M_{i}\}_{i=1}^{n}$ with 
\begin{equation}\label{eq:inmg}
\sum_{i=1}^{n}\eta_{i}\Tr(\rho_{i}M_{i})>p_{\rm SEP}(\mathcal{E}), 
\end{equation}
and this contradicts the optimality of $p_{\rm SEP}(\mathcal{E})$ in Eq.~\eref{eq:pptdef}.
Here, $0_{\mathbb{H}}$ is the zero operator in $\mathbb{H}$.
We also note that the Cartesian product $\mathbb{R}\times\mathbb{H}$
can be considered as a real vector space with an inner product defined as
\begin{equation}\label{eq:inpde}
\langle (a,A),(b,B)\rangle=ab+\Tr(AB)
\end{equation}
for $(a,A),(b,B)\in\mathbb{R}\times\mathbb{H}$.
%%%%%%%%%%%%%%%%%%%%%%

%%%%%%%%%%%%%%%%%%%%%%
Since $\mathcal{S}(\mathcal{E})$ and the single-element set $\{(0,0_{\mathbb{H}})\}$ are disjoint convex sets,
it follows from separating hyperplane theorem\cite{boyd2004,sht} that 
there is $(\gamma,\Gamma)\in\mathbb{R}\times\mathbb{H}$ satisfying
\begin{gather}
(\gamma,\Gamma)\neq(0,0_{\mathbb{H}}),\label{eq:gcgneq}\\
\langle(\gamma,\Gamma),(r,G)\rangle\leqslant0~~\forall(r,G)\in\mathcal{S}(\mathcal{E}).
\label{eq:rcgleq}
\end{gather}
%%%%%%%%%%%%%%%%%%%%%%

%%%%%%%%%%%%%%%%%%%%%%
\indent Suppose 
\begin{gather}
\Tr\Gamma\leqslant\gamma p_{\rm SEP}(\mathcal{E}),\label{eq:sca1}\\
\Gamma-\gamma\eta_{i}\rho_{i}\in\mathbb{SEP}^{*}~\forall i=1,\ldots,n,\label{eq:sca2}\\
\gamma>0.\label{eq:sca3}
\end{gather}
From Conditions~\eref{eq:sca2} and \eref{eq:sca3}, the Hermitian operator
$H=\Gamma/\gamma$ is an element of $\mathbb{H}_{\rm SEP}(\mathcal{E})$ in Eq.~\eref{eq:ses}.
Thus, the definition of $q_{\rm SEP}(\mathcal{E})$ in Eq.~\eref{eq:qptdef} leads us to
\begin{equation}\label{eq:qptlet}
q_{\rm SEP}(\mathcal{E})\leqslant\Tr H.
\end{equation}
Moreover, Condition~\eref{eq:sca1} implies
\begin{equation}\label{eq:thleqp}
\Tr H\leqslant p_{\rm SEP}(\mathcal{E}).
\end{equation}
Inequalities~\eref{eq:qptlet} and \eref{eq:thleqp} imply Inequality~\eref{eq:pqgt}.
%%%%%%%%%%%%%%%%%%%%%%

%%%%%%%%%%%%%%%%%%%%%%
To complete the proof, we show the validity of
\eref{eq:sca1}, \eref{eq:sca2} and \eref{eq:sca3}.
%%%%%%%%%%%%%%%%%%%%%%

%%%%%%%%%%%%%%%%%%%%%%
\noindent\textit{Proof of \eref{eq:sca1}.}
From Eq.~\eref{eq:inpde},
Inequality~\eref{eq:rcgleq} can be rewritten as
\begin{equation}\label{eq:trgmle}
\Tr \Gamma-\sum_{i=1}^{n}\Tr[M_{i}(\Gamma-\gamma\eta_{i}\rho_{i})]\leqslant \gamma p 
\end{equation}
for all $p>p_{\rm SEP}(\mathcal{E})$ and all $\{M_{i}\}_{i=1}^{n}\subseteq\mathbb{SEP}$. 
If $M_{i}=0_{\mathbb{H}}$ for all $i=1,\ldots,n$,  
Inequality~\eref{eq:trgmle} becomes Inequality~\eref{eq:sca1}
by taking the limit of $p$ to $p_{\rm SEP}(\mathcal{E})$.
%%%%%%%%%%%%%%%%%%%%%%

%%%%%%%%%%%%%%%%%%%%%%
\noindent\textit{Proof of \eref{eq:sca2}.}
For each $j\in\{1,\ldots,n\}$,
let us consider an arbitrary $M_{j}\in\mathbb{SEP}$ 
and $M_{i}=0_{\mathbb{H}}$ for all $i=1,\ldots,n$ with $i\neq j$. 
In this case,
$\{M_{i}\}_{i=1}^{n}$ is clearly a subset of $\mathbb{SEP}$,
and Inequality~\eref{eq:trgmle} becomes
\begin{equation}\label{eq:trmjle}
\Tr\Gamma-\Tr[M_{j}(\Gamma-\gamma\eta_{j}\rho_{j})]\leqslant\gamma p_{\rm SEP}(\mathcal{E})
\end{equation}
by taking the limit of $p$ to $p_{\rm SEP}(\mathcal{E})$.
%%%%%%%%%%%%%%%%%%%%%%

%%%%%%%%%%%%%%%%%%%%%%
Suppose $\Gamma-\gamma\eta_{j}\rho_{j}\notin\mathbb{SEP}^{*}$,
then there is $M\in\mathbb{SEP}$ with $\Tr[M(\Gamma-\gamma\eta_{j}\rho_{j})]<0$. 
We note that $M\in\mathbb{SEP}$ implies $tM\in\mathbb{SEP}$ for any $t>0$. 
Thus, $\{M_{i}\}_{i=1}^{n}$ consisting of 
$M_{j}=tM$ for $t>0$ and $M_{i}=0$ for all $i=1,\ldots,n$ with $i\neq j$ 
is also a subset of $\mathbb{SEP}$.
%%%%%%%%%%%%%%%%%%%%%%

%%%%%%%%%%%%%%%%%%%%%%
Now, Inequality~\eref{eq:trmjle}
can be rewritten as
\begin{equation}\label{eq:trtmle}
\Tr\Gamma-\Tr[tM(\Gamma-\gamma\eta_{j}\rho_{j})]\leqslant\gamma p_{\rm SEP}(\mathcal{E}).
\end{equation}
Since Inequality~\eref{eq:trtmle} is true for arbitrary large $t>0$, 
$\gamma p_{\rm SEP}(\mathcal{E})$ can also be arbitrary large.
However, this contradicts that both  $\gamma$ and $p_{\rm SEP}(\mathcal{E})$ are finite.
Thus, $\Gamma-\gamma\eta_{j}\rho_{j}\in\mathbb{SEP}^{*}$, 
which completes the proof of \eref{eq:sca2}.
%%%%%%%%%%%%%%%%%%%%%%

%%%%%%%%%%%%%%%%%%%%%%
\noindent\textit{Proof of \eref{eq:sca3}.}
Suppose $\gamma<0$ and consider $\{M_{i}\}_{i=1}^{n}$ with
$M_{i}=0_{\mathbb{H}}$ for all $i=1,\ldots,n$.
Since $\{M_{i}\}_{i=1}^{n}\subseteq\mathbb{SEP}$,
Inequality~\eref{eq:trgmle} becomes
\begin{equation}\label{eq:tgmi}
\Tr\Gamma\leqslant-\infty
\end{equation}
by taking the limit of $p$ to $\infty$.
This contradicts that $\Gamma$ is bounded, therefore $\gamma\geqslant0$.
%%%%%%%%%%%%%%%%%%%%%%

%%%%%%%%%%%%%%%%%%%%%%
Now, let us suppose $\gamma=0$.
In this case, Conditions \eref{eq:sca1} and \eref{eq:sca2} become
\begin{equation}\label{eq:tglpp}
\Tr\Gamma\leqslant0,~~\Gamma\in\mathbb{SEP}^{*}.
\end{equation}
From Condition~\eref{eq:tglpp} and the argument with Inequality~\eref{eq:trez}, we have 
\begin{equation}\label{eq:gzmh}
\Gamma=0_{\mathbb{H}},
\end{equation}
which contradicts Condition~\eref{eq:gcgneq}. Thus, $\gamma>0$.
\end{proof}
%%%%%%%%%%%%%%%%%%%%%%

%%%%%%%%%%%%%%%%%%%%%%
\indent For a given ensemble $\mathcal{E}=\{\eta_{i},\rho_{i}\}_{i=1}^{n}$, the following theorem provides a necessary and sufficient condition for a separable measurement $\{M_{i}\}_{i=1}^{n}$ and $H\in\mathbb{H}_{\rm SEP}(\mathcal{E})$ to realize $p_{\rm SEP}(\mathcal{E})$ and $q_{\rm SEP}(\mathcal{E})$, respectively.
%%%%%%%%%%%%%%%%%%%%%%

%%%%%%%%%%%%%%%%%%%%%%
%      Theorem       %
%%%%%%%%%%%%%%%%%%%%%%
\begin{theorem}\label{thm:mnsc}
For a multipartite quantum state ensemble $\mathcal{E}=\{\eta_{i},\rho_{i}\}_{i=1}^{n}$, a separable measurement $\{M_{i}\}_{i=1}^{n}$ and $H\in\mathbb{H}_{\rm SEP}(\mathcal{E})$, $\{M_{i}\}_{i=1}^{n}$ realizes $p_{\rm SEP}(\mathcal{E})$ and $H$ provides $q_{\rm SEP}(\mathcal{E})$ if and only if
\begin{equation}\label{eq:comc}
\Tr[M_{i}(H-\eta_{i}\rho_{i})]=0~~\forall i=1,\ldots,n.
\end{equation}
\end{theorem}
%%%%%%%%%%%%%%%%%%%%%%

%%%%%%%%%%%%%%%%%%%%%%
\begin{proof}
Let us suppose that $\{M\}_{i=1}^{n}$ and $H$ 
give $p_{\rm SEP}(\mathcal{E})$ and $q_{\rm SEP}(\mathcal{E})$, respectively.
Because $M_{i}\in\mathbb{SEP}$ and $H-\eta_{i}\rho_{i}\in\mathbb{SEP}^{*}$ for all $i=1,\ldots,n$, we have
\begin{equation}\label{eq:mihige}
\Tr[M_{i}(H-\eta_{i}\rho_{i})]\geqslant0~\forall i=1,\ldots,n.
\end{equation}
Moreover, we have
\begin{equation}\label{eq:stmh}
\sum_{i=1}^{n}\Tr[M_{i}(H-\eta_{i}\rho_{i})]
=\Tr H-\sum_{i=1}^{n}\eta_{i}\Tr(\rho_{i}M_{i})
=q_{\rm SEP}(\mathcal{E})-p_{\rm SEP}(\mathcal{E})=0,
\end{equation}
where the first equality is from $\sum_{i=1}^{n}M_{i}=\mathbbm{1}$,
the second equality is due to the assumption of $H$ and $\{M_{i}\}_{i=1}^{n}$, 
and the last equality is by Theorem~\ref{thm:pptq}. 
Inequality~\eref{eq:mihige} and Eq.~\eref{eq:stmh}
lead us to Condition~\eref{eq:comc}.
%%%%%%%%%%%%%%%%%%%%%%

%%%%%%%%%%%%%%%%%%%%%%
Conversely, let us assume that 
$\{M_{i}\}_{i=1}^{n}$ and $H$ satisfy Condition~\eref{eq:comc}.
This assumption implies
\begin{equation}\label{eq:qpthq}
q_{\rm SEP}(\mathcal{E})
=p_{\rm SEP}(\mathcal{E})\geqslant\sum_{i=1}^{n}\eta_{i}\mathrm{Tr}(\rho_{i}M_{i})
=\sum_{i=1}^{n}\eta_{i}\mathrm{Tr}(\rho_{i}M_{i})
+\sum_{i=1}^{n}\mathrm{Tr}[M_{i}(H-\eta_{i}\rho_{i})]
=\mathrm{Tr}H
\geqslant q_{\rm SEP}(\mathcal{E}),
\end{equation}
where the first equality is by Theorem~\ref{thm:pptq},
the second equality is from Condition~\eref{eq:comc},
the last equality follows from $\sum_{i=1}^{n}M_{i}=\mathbbm{1}$, 
and the first and second inequalities are due to
the definitions of $p_{\rm SEP}(\mathcal{E})$ and $q_{\rm SEP}(\mathcal{E})$, respectively.
Inequality~\eref{eq:qpthq} leads us to 
\begin{equation}\label{eq:semp}
\sum_{i=1}^{n}\eta_{i}\mathrm{Tr}(\rho_{i}M_{i})=p_{\rm SEP}(\mathcal{E}),~~
\Tr H=q_{\rm SEP}(\mathcal{E}). 
\end{equation}
Thus, $\{M_{i}\}_{i=1}^{n}$ and $H$ provide $p_{\rm SEP}(\mathcal{E})$ and $q_{\rm SEP}(\mathcal{E})$, respectively.
\end{proof}
%%%%%%%%%%%%%%%%%%%%%%

%%%%%%%%%%%%%%%%%%%%%%
We note that $H\in\mathbb{H}_{\rm SEP}(\mathcal{E})$ providing $q_{\rm SEP}(\mathcal{E})$ is generally not unique (see Example~2 in \Sref{subsec:nsc}). 
However, the following corollary states the case that 
$H\in\mathbb{H}_{\rm SEP}(\mathcal{E})$ providing $q_{\rm SEP}(\mathcal{E})$ is unique.
%%%%%%%%%%%%%%%%%%%%%%

%%%%%%%%%%%%%%%%%%%%%%
%     Corollary      %
%%%%%%%%%%%%%%%%%%%%%%
\begin{corollary}\label{cor:exer1}
For a multipartite quantum state ensemble $\mathcal{E}=\{\eta_{i},\rho_{i}\}_{i=1}^{n}$, we have
\begin{equation}\label{eq:ppte1}
p_{\rm SEP}(\mathcal{E})=\eta_{1},
\end{equation}
if and only if 
\begin{equation}\label{eq:exppt}
\eta_{1}\rho_{1}-\eta_{i}\rho_{i}\in\mathbb{SEP}^{*}~~\forall i=2,\ldots,n.
\end{equation}
In this case, $\eta_{1}\rho_{1}$ is the only element of $\mathbb{H}_{\rm SEP}(\mathcal{E})$ providing $q_{\rm SEP}(\mathcal{E})$.
\end{corollary}
%%%%%%%%%%%%%%%%%%%%%%

%%%%%%%%%%%%%%%%%%%%%%
\begin{proof}
Let $\{M_{i}\}_{i=1}^{n}$ be the measurement such that 
\begin{equation}\label{eq:mzo}
M_{1}=\mathbbm{1},~~
M_{2}=\cdots=M_{n}=0_{\mathbb{H}}. 
\end{equation}
We first assume Eq.~\eref{eq:ppte1} and consider $H\in\mathbb{H}_{\rm SEP}(\mathcal{E})$ giving $q_{\rm SEP}(\mathcal{E})$.
Since $\{M_{i}\}_{i=1}^{n}$ is obviously a separable measurement providing $p_{\rm SEP}(\mathcal{E})$, it follows from Theorem~\ref{thm:mnsc} that $\Tr(H-\eta_{1}\rho_{1})=0$.
From $H-\eta_{1}\rho_{1}\in\mathbb{SEP}^{*}$ and the argument with Inequality~\eref{eq:trez}, we have $H=\eta_{1}\rho_{1}$.
Thus, $H\in\mathbb{H}_{\rm SEP}(\mathcal{E})$ together with the definition of $\mathbb{H}_{\rm SEP}(\mathcal{E})$ leads us to Condition~\eref{eq:exppt}.
%%%%%%%%%%%%%%%%%%%%%%

%%%%%%%%%%%%%%%%%%%%%%
Conversely, let us suppose Condition~\eref{eq:exppt} and
consider 
\begin{equation}\label{eq:che1}
H=\eta_{1}\rho_{1},
\end{equation}
which is in $\mathbb{H}_{\rm SEP}(\mathcal{E})$ by Condition~\eref{eq:exppt}.
The separable measurement $\{M_{i}\}_{i=1}^{n}$ in Eq.~\eref{eq:mzo} and $H\in\mathbb{H}_{\rm SEP}(\mathcal{E})$ in Eq.~\eref{eq:che1} satisfy Condition~\eref{eq:comc}. Therefore, we have
\begin{equation}\label{eq:seqh1}
p_{\rm SEP}(\mathcal{E})=\sum_{i=1}^{n}\eta_{i}\Tr(\rho_{i}M_{i})=\eta_{1},
\end{equation}
where the first equality is by Theorem~\ref{thm:mnsc}.
\end{proof}
%%%%%%%%%%%%%%%%%%%%%%

%%%%%%%%%%%%%%%%%%%%%%
When Eq.~\eref{eq:ppte1} of Corollary~\ref{cor:exer1} holds, the maximum success probability $p_{\rm SEP}(\mathcal{E})$ can be achieved without the help of measurement,
simply by guessing $\rho_{1}$ is prepared.
As we can check in the proof of Corollary~\ref{cor:exer1}, the choice of $\rho_{1}$ in Corollary~\ref{cor:exer1} can be arbitrary. That is, any of $\{\rho_{i}\}_{i=1}^{n}$ can be used to play the role of $\rho_{1}$ in Corollary~\ref{cor:exer1}.
%%%%%%%%%%%%%%%%%%%%%%

%%%%%%%%%%%%%%%%%%%%%%
%      Section       %
%%%%%%%%%%%%%%%%%%%%%%
\section{Minimum-error discrimination by separable measurements}\label{sec:mep}
%%%%%%%%%%%%%%%%%%%%%%
\indent For a quantum state ensemble $\mathcal{E}$ in Eq.~\eref{eq:ens}, the minimum-error discrimination can be realized by separable measurements if and only if
\begin{equation}\label{eq:inptpg}
p_{\rm SEP}(\mathcal{E})=p_{\rm G}(\mathcal{E}),
\end{equation}
where $p_{\rm G}(\mathcal{E})$ and $p_{\rm SEP}(\mathcal{E})$ are defined in Eqs.~\eref{eq:pgdef} and \eref{eq:pptdef}, respectively. 
In this section, we provide some necessary and/or sufficient conditions for Eq.~\eref{eq:inptpg} in terms of EW.
%%%%%%%%%%%%%%%%%%%%%%

%%%%%%%%%%%%%%%%%%%%%%
%     Subsection     %
%%%%%%%%%%%%%%%%%%%%%%
\subsection{Necessary condition for $p_{\rm SEP}(\mathcal{E})=p_{\rm G}(\mathcal{E})$}
%%%%%%%%%%%%%%%%%%%%%%

%%%%%%%%%%%%%%%%%%%%%%
%      Theorem       %
%%%%%%%%%%%%%%%%%%%%%%
\begin{theorem}\label{thm:qmsc}
For a multipartite quantum state ensemble $\mathcal{E}=\{\eta_{i},\rho_{i}\}_{i=1}^{n}$,
if there exists separable measurement $\{M_{i}\}_{i=1}^{n}$ satisfying
\begin{equation}\label{eq:dewc}
\sum_{i=1}^{n}\eta_{i}\rho_{i}M_{i}\in\mathbb{H}_{\rm EW}(\mathcal{E}),
\end{equation}
where $\mathbb{H}_{\rm EW}(\mathcal{E})$ is defined in Eq.~\eref{eq:hppte},
then 
\begin{equation}\label{eq:ppes}
p_{\rm SEP}(\mathcal{E})=\sum_{i=1}^{n}\eta_{i}\Tr(\rho_{i}M_{i})<p_{\rm G}(\mathcal{E}).
\end{equation}
\end{theorem}
%%%%%%%%%%%%%%%%%%%%%%
\noindent Thus, non-existence of such separable measurement $\{M_{i}\}_{i=1}^{n}$ satisfying Condition~\eref{eq:dewc} is a necessary condition for $p_{\rm SEP}(\mathcal{E})=p_{\rm G}(\mathcal{E})$.
%%%%%%%%%%%%%%%%%%%%%%
\begin{proof}
Let us suppose $\{M_{i}\}_{i=1}^{n}$ is a separable measurement satisfying Condition~\eref{eq:dewc}, and consider
\begin{equation}\label{eq:hsrm}
H=\sum_{i=1}^{n}\eta_{i}\rho_{i}M_{i}.
\end{equation}
\Eref{eq:ppes} holds because
\begin{equation}\label{eq:stth}
\sum_{i=1}^{n}\eta_{i}\Tr(\rho_{i}M_{i})
\leqslant p_{\rm SEP}(\mathcal{E})=q_{\rm SEP}(\mathcal{E})
\leqslant \Tr H=\sum_{i=1}^{n}\eta_{i}\Tr(\rho_{i}M_{i}),
\end{equation}
where 
the first inequality is due to the definition of $p_{\rm SEP}(\mathcal{E})$, 
the second inequality follows from the definition of $q_{\rm SEP}(\mathcal{E})$ together with the condition that $H\in\mathbb{H}_{\rm EW}(\mathcal{E})\subseteq\mathbb{H}_{\rm SEP}(\mathcal{E})$,
the first equality is from Theorem~\ref{thm:pptq}, and the second equality is by Eq.~\eref{eq:hsrm}. This proves the equality of \eref{eq:ppes}.
%%%%%%%%%%%%%%%%%%%%%%

%%%%%%%%%%%%%%%%%%%%%%
If we assume $p_{\rm SEP}(\mathcal{E})=p_{\rm G}(\mathcal{E})$
in \eref{eq:ppes}, $\{M_{i}\}_{i=1}^{n}$ gives the optimal success probability $p_{\rm G}(\mathcal{E})$.
From the optimality condition in Eq.~\eref{eq:nscfme} and the argument with Eq.~\eref{eq:psmd},
we have
\begin{equation}\label{eq:ser}
\sum_{i=1}^{n}\eta_{i}\rho_{i}M_{i}\in\mathbb{H}_{\rm SEP}(\mathcal{E})\setminus\mathbb{H}_{\rm EW}(\mathcal{E}), 
\end{equation}
which contradicts Condition~\eref{eq:dewc}.
Thus, we have the inequality of \eref{eq:ppes}.
\end{proof}
%%%%%%%%%%%%%%%%%%%%%%%

%%%%%%%%%%%%%%%%%%%%%%
%       Example      %
%%%%%%%%%%%%%%%%%%%%%%
\noindent
\textbf{Example 1.}
For any integers $m,d\geqslant 2$, let us consider the $m$-qu$d$it state ensemble $\mathcal{E}=\{\eta_{i},\rho_{i}\}_{i=1}^{d+2}$ consisting of $d+2$ states,
\begin{flalign}\label{eq:exes}
&\eta_{i}=\frac{1}{d^{m}+d},~\rho_{i}=\ket{i-1}\!\bra{i-1}^{\otimes m},~i=1,\ldots,d,\nonumber\\
&\eta_{d+1}=\frac{d^{m}-d}{d^{m}+d},~\rho_{d+1}=\frac{1}{d^{m}-d}\Big(\mathbbm{1}-\sum_{j=0}^{d-1}\ket{j}\!\bra{j}^{\otimes m}\Big),\\
&\eta_{d+2}=\frac{d}{d^{m}+d},~\rho_{d+2}=\ket{\Phi}\!\bra{\Phi},\nonumber
\end{flalign}
where 
\begin{equation}\label{eq:psijk}
\ket{\Phi}=\frac{1}{\sqrt{d}}\sum_{i=0}^{d-1}\ket{i}^{\otimes m}.
\end{equation}
For a separable measurement $\{M_{i}\}_{i=1}^{d+2}$ with
\begin{equation}\label{eq:mij}
M_{i}=\ket{i-1}\!\bra{i-1}^{\otimes m},~i=1,\ldots,d,~
M_{d+1}=\mathbbm{1}-\sum_{j=0}^{d-1}\ket{j}\!\bra{j}^{\otimes m},~
M_{d+2}=0_{\mathbb{H}},
\end{equation}
we show that Condition~\eref{eq:dewc} holds with respect to the ensemble in Eq.~\eref{eq:exes}.
%%%%%%%%%%%%%%%%%%%%%%

%%%%%%%%%%%%%%%%%%%%%%
It is straightforward to verify that
\begin{flalign}\label{eq:sclt1}
&\sum_{j=1}^{d+2}\eta_{j}\rho_{j}M_{j}-\eta_{i}\rho_{i}=\frac{1}{d^{m}+d}\big(\mathbbm{1}-\ket{i-1}\!\bra{i-1}^{\otimes m}\big)\in\mathbb{H}_{+},~i=1,\ldots,d,\nonumber\\
&\sum_{j=1}^{d+2}\eta_{j}\rho_{j}M_{j}-\eta_{d+1}\rho_{d+1}=\frac{1}{d^{m}+d}\sum_{j=0}^{d-1}\ket{j}\!\bra{j}^{\otimes m}\in\mathbb{H}_{+},\\
&\sum_{j=1}^{d+2}\eta_{j}\rho_{j}M_{j}-\eta_{d+2}\rho_{d+2}=\frac{1}{d^{m}+d}\big(\mathbbm{1}-d\ket{\Phi}\!\bra{\Phi}\big)\in\mathbb{SEP}^{*},\nonumber
\end{flalign}
where the last inclusion is from the fact that
\begin{equation}\label{eq:dtle}
d\Tr(\ket{\Phi}\!\bra{\Phi}E)\leqslant\Tr E~~\forall E\in\mathbb{SEP}.
\end{equation}
To show the validity of Inequality~\eref{eq:dtle}, we assume that $\{\ket{e_{i}^{(k)}}\}_{i=1}^{d}$ is an orthonormal basis of $\mathbb{C}^{d}$ for each $k=1,\ldots,m$. Since $\{\ket{e_{i_{1}}^{(1)}}\otimes\cdots\otimes\ket{e_{i_{m}}^{(m)}}\}_{i_{1},\ldots,i_{m}}$ is an orthonormal basis of the multipartite Hilbert space $\mathcal{H}$,
there is a set of complex numbers $\{c_{i_{1},\ldots,i_{m}}\}_{i_{1},\ldots,i_{m}}$ such that
\begin{equation}\label{eq:phrp}
\ket{\Phi}=\sum_{i_{1},\ldots,i_{m}=1}^{d}c_{i_{1},\ldots,i_{m}}\ket{e_{i_{1}}^{(1)}}\otimes\cdots\otimes\ket{e_{i_{m}}^{(m)}}.
\end{equation}
For each $i_{1}\in\{1,\ldots,d\}$, we have
\begin{equation}\label{eq:odte}
\sum_{i_{2},\ldots,i_{m}=1}^{d}|c_{i_{1},\ldots,i_{m}}|^{2}=
\bra{\Phi}\Big[\ket{e_{i_{1}}^{(1)}}\!\bra{e_{i_{1}}^{(1)}}\otimes\big(\sum_{j=0}^{d-1}\ket{j}\!\bra{j}\big)^{\otimes m-1}\Big]\ket{\Phi}
=\frac{1}{d},
\end{equation}
where the first equality follows from Eq.~\eref{eq:phrp} and the second equality is due to Eq.~\eref{eq:psijk}. Equation~\eref{eq:odte} leads us to
\begin{equation}\label{eq:cimd}
|c_{i_{1},\ldots,i_{m}}|^{2}\leqslant\frac{1}{d}~~\forall i_{1},\ldots,i_{m}=1,\ldots,d.
\end{equation}
Since the choice of $\{\ket{e_{i}^{(k)}}\}_{i=1}^{d}$ can be arbitrary for each $k=1,\ldots,m$, we have
\begin{equation}
d\Tr(|\Phi\rangle\!\langle\Phi|e\rangle\!\langle e|)\leqslant \langle e|e\rangle
\end{equation}
for any product vector $\ket{e}\in\mathcal{H}$. Therefore, Inequality~\eref{eq:dtle} holds.
%%%%%%%%%%%%%%%%%%%%%%

%%%%%%%%%%%%%%%%%%%%%
Now, the inclusions in \eref{eq:sclt1} together with $\mathbb{H}_{+}\subseteq\mathbb{SEP}^{*}$ imply
\begin{equation}\label{eq:ipml}
\sum_{j=1}^{d+2}\eta_{j}\rho_{j}M_{j}-\eta_{i}\rho_{i}\in\mathbb{SEP}^{*}~~\forall i=1,\ldots,d+2.
\end{equation}
Furthermore, a straightforward calculation leads us to
\begin{equation}\label{eq:mwhe}
\bra{\Phi}\Big(\sum_{j=1}^{d+2}\eta_{j}\rho_{j}M_{j}-\eta_{d+2}\rho_{d+2}\Big)\ket{\Phi}
=-\frac{d-1}{d^{m}+d}<0.
\end{equation}
From Eq.~\eref{eq:mwhe}, we have
\begin{equation}\label{eq:mkle}
\sum_{j=1}^{d+2}\eta_{j}\rho_{j}M_{j}-\eta_{d+2}\rho_{d+2}\notin\mathbb{H}_{+}.
\end{equation}
From Eqs.~\eref{eq:ipml} and \eref{eq:mkle},
the ensemble in Eq.~\eref{eq:exes} and the measurement in Eq.~\eref{eq:mij} satisfy
Condition~\eref{eq:dewc}. Thus, Theorem~\ref{thm:qmsc} leads us to
\begin{equation}\label{eq:pse6}
p_{\rm SEP}(\mathcal{E})
=\sum_{i=1}^{d+2}\eta_{i}\Tr(\rho_{i}M_{i})=\frac{d^{m}}{d^{m}+d}<p_{\rm G}(\mathcal{E}).
\end{equation}
%%%%%%%%%%%%%%%%%%%%%%

%%%%%%%%%%%%%%%%%%%%%%
We also note that the separable measurement in Eq.~\eref{eq:mij} is a LOCC measurement because it can be implemented by performing the same local measurement $\{\ket{l}\!\bra{l}\}_{l=0}^{d-1}$ on each party. Thus, we have
\begin{equation}\label{eq:plep}
p_{\rm L}(\mathcal{E})=p_{\rm SEP}(\mathcal{E})=\frac{d^{m}}{d^{m}+d}.
\end{equation}
%%%%%%%%%%%%%%%%%%%%%%

%%%%%%%%%%%%%%%%%%%%%%
%     Subsection     %
%%%%%%%%%%%%%%%%%%%%%%
\subsection{Necessary and sufficient condition for $p_{\rm SEP}(\mathcal{E})=p_{\rm G}(\mathcal{E})$}\label{subsec:nsc}
%%%%%%%%%%%%%%%%%%%%%%

%%%%%%%%%%%%%%%%%%%%%%
%      Theorem       %
%%%%%%%%%%%%%%%%%%%%%%
\begin{theorem}\label{thm:qupb}
For a multipartite quantum state ensemble $\mathcal{E}=\{\eta_{i},\rho_{i}\}_{i=1}^{n}$, $p_{\rm SEP}(\mathcal{E})=p_{\rm G}(\mathcal{E})$ if and only if
there exists $H\in\mathbb{H}_{\rm SEP}(\mathcal{E})$
such that it provides $q_{\rm SEP}(\mathcal{E})$
but does not satisfy
\begin{equation}\label{eq:hdc}
H\in\mathbb{H}_{\rm EW}(\mathcal{E}),
\end{equation}
or equivalently, there is $H\in\mathbb{H}$ satisfying Condition~\eref{eq:psmd} and $\Tr H=q_{\rm SEP}(\mathcal{E})$.
\end{theorem}
%%%%%%%%%%%%%%%%%%%%%%

%%%%%%%%%%%%%%%%%%%%%%
\begin{proof}
Let $\{M_{i}\}_{i=1}^{n}$ be a separable measurement giving $p_{\rm SEP}(\mathcal{E})$. We first suppose $p_{\rm SEP}(\mathcal{E})=p_{\rm G}(\mathcal{E})$ and consider 
\begin{equation}\label{eq:hjn}
H=\sum_{i=1}^{n}\eta_{i}\rho_{i}M_{i}.
\end{equation}
Since the measurement $\{M_{i}\}_{i=1}^{n}$ gives
the optimal success probability $p_{\rm G}(\mathcal{E})$,
the optimality condition in Eq.~\eref{eq:nscfme} leads us to
\begin{equation}
H-\eta_{i}\rho_{i}
=\sum_{j=1}^{n}\eta_{j}\rho_{j}M_{j}-\eta_{i}\rho_{i}\in\mathbb{H}_{+}~~\forall i=1,\ldots,n.
\label{eq:hemp}
\end{equation}
Therefore, $H$ satisfies Condition~\eref{eq:psmd}. Moreover, we have
\begin{equation}\label{eq:trqpt}
\Tr H=\sum_{i=1}^{n}\eta_{i}\Tr(\rho_{i}M_{i})=p_{\rm SEP}(\mathcal{E})=q_{\rm SEP}(\mathcal{E}),
\end{equation}
where the first equality is from Eq.~\eref{eq:hjn},
the second equality is by the assumption of $\{M_{i}\}_{i=1}^{n}$, and 
the last equality is due to Theorem~\ref{thm:pptq}.
%%%%%%%%%%%%%%%%%%%%%%

%%%%%%%%%%%%%%%%%%%%%%
Conversely, let us assume $H$ is an element of $\mathbb{H}$ satisfying Condition~\eref{eq:psmd} and $\Tr H=q_{\rm SEP}(\mathcal{E})$.
Thus, the positivenesses in \eref{eq:eqhpt} is satisfied in term of $H$.
For each $i=1,\ldots,n$, the positive-semidefinite operators $M_{i}$ and $H-\eta_{i}\rho_{i}$ are orthogonal since they satisfy Condition~\eref{eq:comc} from Theorem~\ref{thm:mnsc}.
%%%%%%%%%%%%%%%%%%%%%%

%%%%%%%%%%%%%%%%%%%%%%
The optimality condition in Eq.~\eref{eq:nscfme} holds for the measurement $\{M_{i}\}_{i=1}^{n}$ because
\begin{equation}\label{eq:hetai}
\sum_{j=1}^{n}\eta_{j}\rho_{j}M_{j}-\eta_{i}\rho_{i}
=\sum_{j=1}^{n}\eta_{j}\rho_{j}M_{j}
+\sum_{k=1}^{n}(H-\eta_{k}\rho_{k})M_{k}-\eta_{i}\rho_{i}
=H-\eta_{i}\rho_{i}\in\mathbb{H}_{+}~~\forall i=1,\ldots,n,
\end{equation}
where the first equality is from the orthogonality of $M_{i}$ and $H-\eta_{i}\rho_{i}$ for each $i=1,\ldots,n$ and the second equality is from $\sum_{i=1}^{n}M_{i}=\mathbbm{1}$.
Thus, we have
\begin{equation}\label{eq:pgeps}
p_{\rm G}(\mathcal{E})
=\sum_{i=1}^{n}\eta_{i}\Tr(\rho_{i}M_{i})
=p_{\rm SEP}(\mathcal{E}),
\end{equation}
where the second equality is due to the assumption of $\{M_{i}\}_{i=1}^{n}$.
\end{proof}
%%%%%%%%%%%%%%%%%%%%%%

%%%%%%%%%%%%%%%%%%%%%%
If $p_{\rm SEP}(\mathcal{E})=p_{\rm G}(\mathcal{E})$,
Theorem~\ref{thm:qupb} implies that there must exist $H\in\mathbb{H}_{\rm SEP}(\mathcal{E})\setminus\mathbb{H}_{\rm EW}(\mathcal{E})$ providing $q_{\rm SEP}(\mathcal{E})$.
In this case, there possibly exists another Hermitian operator $H'$ satisfying
$H'\in\mathbb{H}_{\rm EW}(\mathcal{E})$ and $\Tr H'= q_{\rm SEP}(\mathcal{E})$, 
which is illustrated in the following example.
%%%%%%%%%%%%%%%%%%%%%%

%%%%%%%%%%%%%%%%%%%%%%
%       Example      %
%%%%%%%%%%%%%%%%%%%%%%
\noindent
\textbf{Example 2.}
For any integers $m,d\geqslant 2$, let us consider the $m$-qu$d$it state ensemble $\mathcal{E}=\{\eta_{i},\rho_{i}\}_{i=1}^{d^{m}-d+1}$ consisting of $d^{m}-d+1$ states,
\begin{equation}\label{eq:tpose}
\eta_{1}=\frac{d}{d^{m}},~\rho_{1}=\ket{\Phi}\!\bra{\Phi},~
\eta_{i}=\frac{1}{d^{m}},~\rho_{i}=\ket{\beta_{i}}\!\bra{\beta_{i}},~i=2,\ldots,d^{m}-d+1,
\end{equation}
where $\ket{\Phi}$ is defined in Eq.~\eref{eq:psijk} and $\{\ket{\beta_{i}}\}_{i=2}^{d^{m}-d+1}$ is a set of orthonormal product vectors orthogonal to $\ket{j}^{\otimes m}$ for all $j=0,\ldots,d-1$.
%%%%%%%%%%%%%%%%%%%%%%

%%%%%%%%%%%%%%%%%%%%%%
For a separable measurement $\{M_{i}\}_{i=1}^{d^{m}-d+1}$ with
\begin{equation}\label{eq:expl2}
M_{1}=\sum_{j=0}^{d-1}\ket{j}\!\bra{j}^{\otimes m},
~M_{i}=\ket{\beta_{i}}\!\bra{\beta_{i}},~i=2,\ldots,d^{m}-d+1,
\end{equation}
we can easily see that the success probability obtained from the separable measurement in discriminating the states from the ensemble $\mathcal{E}$ in Eq.~\eref{eq:tpose} is one, that is,
\begin{equation}\label{eq:esp}
\sum_{i=1}^{d^{m}-d+1}\eta_{i}\Tr(\rho_{i}M_{i})=1.
\end{equation}
The success probability in Eq.~\eref{eq:esp} is a lower bound of $p_{\rm SEP}(\mathcal{E})$ in Eq.~\eref{eq:pptdef}, therefore
\begin{equation}\label{eq:elb}
p_{\rm SEP}(\mathcal{E})\geqslant1.
\end{equation}
Since $p_{\rm G}(\mathcal{E})$ is bounded above by 1, 
Inequalities~\eref{eq:plptpg} and \eref{eq:elb} imply
\begin{equation}\label{eq:allpeq}
p_{\rm SEP}(\mathcal{E})=p_{\rm G}(\mathcal{E})=1.
\end{equation}
Furthermore, we have
\begin{equation}\label{eq:qex2}
q_{\rm SEP}(\mathcal{E})=p_{\rm SEP}(\mathcal{E})=1,
\end{equation}
where the first equality is by Theorem~\ref{thm:pptq}
and the second equality is from Eq.~\eref{eq:allpeq}.
%%%%%%%%%%%%%%%%%%%%%%

%%%%%%%%%%%%%%%%%%%%%%
Let us first consider the Hermitian operator
\begin{equation}\label{eq:hex2}
H=\sum_{i=1}^{d^{m}-d+1}\eta_{i}\rho_{i}.
\end{equation}
Equations~\eref{eq:qex2} and \eref{eq:hex2} imply
\begin{equation}\label{eq:trh1}
\Tr H=1=q_{\rm SEP}(\mathcal{E}).
\end{equation}
Moreover, a straightforward calculation leads us to
\begin{equation}\label{eq:exhe}
H-\eta_{i}\rho_{i}=
\sum_{\substack{j=1\\j\neq i}}^{d^{m}-d+1}\eta_{j}\rho_{j}\in\mathbb{H}_{+}~\forall i=1,\ldots,n.
\end{equation}
Due to Eqs.~\eref{eq:trh1} and \eref{eq:exhe}, $H$ is an element of $\mathbb{H}_{\rm SEP}(\mathcal{E})\setminus\mathbb{H}_{\rm EW}(\mathcal{E})$ giving $q_{\rm SEP}(\mathcal{E})$.
Thus, we have $H$ satisfying the conditions of Theorem~\ref{thm:qupb}.  
%%%%%%%%%%%%%%%%%%%%%%

%%%%%%%%%%%%%%%%%%%%%%
Now, let us consider the Hermitian operator
\begin{equation}\label{eq:hexf2}
H_{t}=t\sum_{i=1}^{d^{m}-d+1}\eta_{i}\rho_{i}+\frac{1-t}{d^{m}}\mathbbm{1},
\end{equation}
where $0\leqslant t<1$.
Equations~\eref{eq:qex2} and \eref{eq:hexf2} imply
\begin{equation}\label{eq:trhe1}
\Tr H_{t}=1=q_{\rm SEP}(\mathcal{E})
\end{equation}
for $0\leqslant t<1$. Moreover,
a straightforward calculation leads us to
\begin{flalign}\label{eq:exhef}
&H_{t}-\eta_{1}\rho_{1}=
t\sum_{j=2}^{d^{m}-d+1}\eta_{j}\rho_{j}+\frac{1-t}{d^{m}}(\mathbbm{1}-d\ket{\Phi}\!\bra{\Phi})
\in\mathbb{SEP}^{*},\nonumber\\
&H_{t}-\eta_{i}\rho_{i}=
t\sum_{\substack{j=1\\j\neq i}}^{d^{m}-d+1}\eta_{j}\rho_{j}+\frac{1-t}{d^{m}}(\mathbbm{1}-\ket{\beta_{i}}\!\bra{\beta_{i}})\in\mathbb{H}_{+}~\forall i\neq1,
\end{flalign}
where the first inclusion follows from Inequality~\eref{eq:dtle}.
Due to Eq.~\eref{eq:exhef} together with $\mathbb{H}_{+}\subseteq\mathbb{SEP}^{*}$, we have
\begin{equation}\label{eq:htsep}
H_{t}\in\mathbb{H}_{\rm SEP}(\mathcal{E}). 
\end{equation}
We also note that
\begin{equation}\label{eq:inpr}
\bra{\Phi}(H_{t}-\eta_{1}\rho_{1})\ket{\Phi}
=-\frac{(1-t)(d-1)}{d^{m}}<0
\end{equation}
for $0\leqslant t<1$. Equations~\eref{eq:trhe1}, \eref{eq:htsep} and Inequality~\eref{eq:inpr} imply that $H_{t}$ is an element of $\mathbb{H}_{\rm EW}(\mathcal{E})$ giving $q_{\rm SEP}(\mathcal{E})$. By letting $H'=H_{t}$, we have another Hermitian operator $H'$, besides $H$ in Eq.~\eref{eq:hex2},
satisfying $H'\in\mathbb{H}_{\rm EW}(\mathcal{E})$ and $\Tr H'= q_{\rm SEP}(\mathcal{E})$.
%%%%%%%%%%%%%%%%%%%%%%

%%%%%%%%%%%%%%%%%%%%%%
For a multipartite quantum state ensemble $\mathcal{E}=\{\eta_{i},\rho_{i}\}_{i=1}^{n}$
where $H$ is the only element of $\mathbb{H}_{\rm SEP}(\mathcal{E})$ providing $q_{\rm SEP}(\mathcal{E})$, Theorem~\ref{thm:qupb} tell us that $p_{\rm SEP}(\mathcal{E})<p_{\rm G}(\mathcal{E})$
if and only if there exists an EW in $\{H-\eta_{i}\rho_{i}\}_{i=1}^{n}$.
From Corollary~\ref{cor:exer1},
$\eta_{1}\rho_{1}$ is the only element of $\mathbb{H}_{\rm SEP}(\mathcal{E})$ providing $q_{\rm SEP}(\mathcal{E})$ when Condition~\eref{eq:exppt} holds.
Thus, we have the following corollary.
%%%%%%%%%%%%%%%%%%%%%%

%%%%%%%%%%%%%%%%%%%%%%
\begin{corollary}\label{cor:pql1}
For a multipartite quantum state ensemble $\mathcal{E}=\{\eta_{i},\rho_{i}\}_{i=1}^{n}$ with Condition~\eref{eq:exppt}, $p_{\rm SEP}(\mathcal{E})<p_{\rm G}(\mathcal{E})$ if and only if 
there exists an EW in $\{\eta_{1}\rho_{1}-\eta_{i}\rho_{i}\}_{i=2}^{n}$.
\end{corollary}
%%%%%%%%%%%%%%%%%%%%%%

%%%%%%%%%%%%%%%%%%%%%%
%       Example      %
%%%%%%%%%%%%%%%%%%%%%%
\noindent
\textbf{Example 3.}
For any integers $m,d\geqslant 2$, let us consider the $m$-qu$d$it state ensemble $\mathcal{E}=\{\eta_{i},\rho_{i}\}_{i=1}^{d+1}$ consisting of $d+1$ states, 
\begin{equation}\label{eq:exerho}
\eta_{1}=\frac{1}{2},~\rho_{1}=\frac{1}{d^{m}}\mathbbm{1},~
\eta_{i}=\frac{1}{2d},~\rho_{i}=\frac{d^{2}-d}{d^{m}-d}\ket{\Phi_{i}}\!\bra{\Phi_{i}}+\frac{d^{m}-d^{2}}{d^{m}(d^{m}-d)}\mathbbm{1},~i=2,\ldots,d+1,
\end{equation}
where 
\begin{equation}\label{eq:pjdkm}
\ket{\Phi_{j}}=\frac{1}{\sqrt{d}}\sum_{k=0}^{d-1}\exp\Big(\frac{{\rm i}2\pi jk}{d}\Big)\ket{k}^{\otimes m}.
\end{equation}
%%%%%%%%%%%%%%%%%%%%%%

%%%%%%%%%%%%%%%%%%%%%%
For each $i=2,\ldots,d+1$,
a straightforward calculation leads us to
\begin{equation}\label{eq:erm0}
\eta_{1}\rho_{1}-\eta_{i}\rho_{i}
=\frac{d-1}{2d(d^{m}-d)}(\mathbbm{1}-d\ket{\Phi_{i}}\!\bra{\Phi_{i}})\in\mathbb{SEP}^{*},
\end{equation}
where the inclusion is from the fact that
\begin{equation}\label{eq:hips}
d\Tr(\ket{\Phi_{i}}\!\bra{\Phi_{i}}E)\leqslant\Tr E~~\forall E\in\mathbb{SEP}.
\end{equation}
We can show the validity of Inequality~\eref{eq:hips} in a similar way to that of Inequality~\eref{eq:dtle}.
%%%%%%%%%%%%%%%%%%%%%%

%%%%%%%%%%%%%%%%%%%%%%
Now, the inclusion in \eref{eq:erm0} together with Corollary~\ref{cor:exer1} imply
\begin{equation}\label{eq:ped}
p_{\rm SEP}(\mathcal{E})=\eta_{1}=\frac{1}{2}.
\end{equation}
Furthermore, a straightforward calculation leads us to
\begin{equation}\label{eq:inr}
\bra{\Phi_{i}}\Big(\eta_{1}\rho_{1}-\eta_{i}\rho_{i}\Big)\ket{\Phi_{i}}
=-\frac{(d-1)^{2}}{2d(d^{m}-d)}<0
~~\forall i=2,\ldots,d+1.
\end{equation}
From Eq.~\eref{eq:inr}, we have
\begin{equation}\label{eq:ermnp}
\eta_{1}\rho_{1}-\eta_{i}\rho_{i}\notin\mathbb{H}_{+}
~~\forall i=2,\ldots,d+1.
\end{equation}
From Eqs.~\eref{eq:erm0} and \eref{eq:ermnp}, 
$\eta_{1}\rho_{1}-\eta_{i}\rho_{i}$ is an EW for any $i=2,\ldots,d+1$.
Thus, Corollary~\ref{cor:pql1} leads us to 
\begin{equation}\label{eq:pshl}
p_{\rm SEP}(\mathcal{E})=\frac{1}{2}<p_{\rm G}(\mathcal{E}).
\end{equation}
%%%%%%%%%%%%%%%%%%%%%%

%%%%%%%%%%%%%%%%%%%%%%
%      Section       %
%%%%%%%%%%%%%%%%%%%%%%
\section{Construction of nonlocal quantum state ensemble}\label{sec:ces}
%%%%%%%%%%%%%%%%%%%%%%
In this section, we provide a systematic way in terms of EW to construct multipartite quantum state ensembles showing nonlocality in state discrimination, that is, $p_{\rm L}(\mathcal{E})<p_{\rm G}(\mathcal{E})$.
For a given EW $W$, let us consider
the multipartite quantum state ensemble $\mathcal{E}=\{\eta_{i},\rho_{i}\}_{i=1}^{2}$ where 
\begin{equation}\label{eq:dfex}
\eta_{1}=\frac{\Tr(P+W)}{\Tr(2P+W)},~
\rho_{1}=\frac{P+W}{\Tr(P+W)},~
\eta_{2}=\frac{\Tr P}{\Tr(2P+W)},~
\rho_{2}=\frac{P}{\Tr P},
\end{equation}
with any $P\in\mathbb{H}_{+}$ satisfying 
\begin{equation}
P+W\in\mathbb{H}_{+}.
\end{equation}
Since $\eta_{1}\rho_{1}-\eta_{2}\rho_{2}$
is proportional to the EW $W$, $p_{\rm SEP}(\mathcal{E})<p_{\rm G}(\mathcal{E})$ holds from
Corollary~\ref{cor:pql1}.
Thus, Inequality~\eref{eq:plptpg} leads us to
$p_{\rm L}(\mathcal{E})<p_{\rm G}(\mathcal{E})$.
%%%%%%%%%%%%%%%%%%%%%%

%%%%%%%%%%%%%%%%%%%%%%
Corollary~\ref{cor:pql1} can also be used to construct 
a multipartite quantum state ensemble $\mathcal{E}=\{\eta_{i},\rho_{i}\}_{i=1}^{n}$ with $n>2$ showing nonlocality in quantum state discrimination. For a set of EWs $\{W_{i}\}_{i=2}^{n}$,
let us consider the multipartite quantum state ensemble $\mathcal{E}=\{\eta_{i},\rho_{i}\}_{i=1}^{n}$ where
\begin{equation}
\eta_{1}=\frac{\Tr\mathbbm{1}}{\Tr(n\mathbbm{1}-\sum_{j=2}^{n}\lambda_{j}W_{j})},~
\rho_{1}=\frac{\mathbbm{1}}{\Tr\mathbbm{1}},~
\eta_{i}=\frac{\Tr(\mathbbm{1}-\lambda_{i}W_{i})}{\Tr(n\mathbbm{1}-\sum_{j=2}^{n}\lambda_{j}W_{j})},~
\rho_{i}=\frac{\mathbbm{1}-\lambda_{i}W_{i}}{\Tr(\mathbbm{1}-\lambda_{i}W_{i})},
~i=2,\ldots,n,
\end{equation}
with any set of positive real numbers $\{\lambda_{i}\}_{i=2}^{n}$ satisfying 
\begin{equation}
\mathbbm{1}-\lambda_{i}W_{i}\in\mathbb{H}_{+}~~\forall i=2,\ldots,n.
\end{equation}
Because $\eta_{1}\rho_{1}-\eta_{i}\rho_{i}$ is proportional to $W_{i}$ for any $i\in\{2,\ldots,n\}$,
$p_{\rm SEP}(\mathcal{E})<p_{\rm G}(\mathcal{E})$ holds from
Corollary~\ref{cor:pql1}.
Thus, Inequality~\eref{eq:plptpg} leads us to
$p_{\rm L}(\mathcal{E})<p_{\rm G}(\mathcal{E})$.
%%%%%%%%%%%%%%%%%%%%%%

%%%%%%%%%%%%%%%%%%%%%%
%    Conclusions     %
%%%%%%%%%%%%%%%%%%%%%%
\section{Conclusions}\label{sec:dis}
%%%%%%%%%%%%%%%%%%%%%%
We have considered multipartite quantum state discrimination and shown that the minimum-error discrimination by separable measurements strongly depends on the existence of EW. We have established the necessary and/or sufficient conditions on minimum-error discrimination by separable measurements, that is, $p_{\rm SEP}(\mathcal{E})=p_{\rm G}(\mathcal{E})$, in terms of EW (Theorems~\ref{thm:qmsc} and \ref{thm:qupb}). We have also provided the conditions on the upper bound of the maximum success probability over all possible separable measurements (Theorems~\ref{thm:pptq} and \ref{thm:mnsc}). 
Our results have been illustrated by examples of multidimensional multipartite quantum states.
Finally, we have provided a systematic way in terms of EW to construct multipartite quantum state ensembles showing nonlocality in state discrimination.
%%%%%%%%%%%%%%%%%%%%%%

%%%%%%%%%%%%%%%%%%%%%%
Quantum nonlocality is a key ingredient making quantum states outperform the classical ones in various quantum information processing tasks such as quantum teleportation and quantum cryptography\cite{eker1991,benn1993}. 
It is also known that quantum nonlocality plays an important role in quantum algorithms 
which are more powerful than any classical ones\cite{deut1992,shor1994}. 
As the violation of the conditions in Theorem~\ref{thm:qupb} implies $p_{\rm SEP}(\mathcal{E})<p_{\rm G}(\mathcal{E})$,
which consequently means $p_{\rm L}(\mathcal{E})<p_{\rm G}(\mathcal{E})$, our results provides a useful methodology to guarantee the occurrence of nonlocality in state discrimination.
%%%%%%%%%%%%%%%%%%%%%%

%%%%%%%%%%%%%%%%%%%%%%
Our results establish a specific relation between the properties of EW and minimum-error discrimination by separable measurements, therefore it is natural to investigate the relationship between EW and other measurements.
It is also an interesting future work to construct good conditions, in terms of EW, for optimal state discrimination in other state discrimination strategies.
%%%%%%%%%%%%%%%%%%%%%%

%%%%%%%%%%%%%%%%%%%%%%
%  Acknowledgments   %
%%%%%%%%%%%%%%%%%%%%%%
\section*{ACKNOWLEDGEMENTS}
\noindent
This work was supported by Basic Science Research Program(NRF-2020R1F1A1A010501270) and Quantum Computing Technology Development Program(NRF-2020M3E4A1080088) through the National Research Foundation of Korea(NRF) grant funded by the Korea government(Ministry of Science and ICT).
%%%%%%%%%%%%%%%%%%%%%%

%%%%%%%%%%%%%%%%%%%%%%
%     References     %
%%%%%%%%%%%%%%%%%%%%%%


\begin{thebibliography}{99}
\bibitem{sht}  
When $A$ and $B$ are disjoint convex sets in a real vector space $V$ with an inner product $\langle\cdot,\cdot\rangle$,
there exist $x\in\mathbb{R}$ and $\vec{v}\in V$ such that $\vec{v}\neq\vec{0}$ and $\langle \vec{a},\vec{v}\rangle\leqslant x\leqslant\langle \vec{b},\vec{v}\rangle$  for all $\vec{a}\in A$ and all $\vec{b}\in B$.
\end{thebibliography}
\end{document}